\documentclass[conference]{IEEEconf}
\IEEEoverridecommandlockouts
\usepackage[noadjust]{cite}
\usepackage{amsmath,amssymb,amsfonts}
\usepackage{algorithmic}
\usepackage{graphicx}
\usepackage{textcomp}
\usepackage{xcolor}
\usepackage{dsfont}
\usepackage{amsthm}
\usepackage{xcolor}
\usepackage{mathtools}

\let \labelindent \relax
\usepackage{enumitem}
\usepackage{subcaption}
\usepackage[font=small]{caption}

\usepackage{accents}

\newtheorem{theorem}{Theorem}[section]

\newtheorem{proposition}[theorem]{Proposition}

\newcommand{\R}{{\mathbb{R}}}

\def\BibTeX{{\rm B\kern-.05em{\sc i\kern-.025em b}\kern-.08em
    T\kern-.1667em\lower.7ex\hbox{E}\kern-.125emX}}
\begin{document}

\title{Mixed-feedback oscillations in the foraging dynamics of arboreal turtle ants
\thanks{}
}

\author{Alia~Valentine, Deborah~M.~Gordon, Anastasia~Bizyaeva% <-this % stops a space
\thanks{A. Valentine is with the Center for Applied Mathematics at Cornell University, Ithaca, NY, 14850, \tt{av589@cornell.edu}}% <-this % stops a space
\thanks{D. M. Gordon is with the Department of Biology at Stanford University, Stanford, CA, 94305, \tt{dmgordon@stanford.edu} }%
\thanks{A. Bizyaeva is with the Sibley School of Mechanical and Aerospace Engineering at Cornell University, Ithaca, NY, 14850, \tt{anastasiab@cornell.edu}}
}%

\maketitle

\begin{abstract} We propose and analyze a model for the dynamics of the flow into and out of a nest for the arboreal turtle ant \textit{Cephalotes goniodontus} during foraging to investigate a possible mechanism for the emergence of oscillations. In our model, there is mixed dynamic feedback between the flow of ants between different behavioral compartments and the amount of pheromone along trails. On one hand, the ants deposit pheromone along the trail, which provides a positive feedback by increasing rates of return to the nest. 
On the other hand, pheromone evaporation is a source of negative feedback, as it depletes the pheromone and inhibits the return rate. We prove that the model is globally asymptotically stable in the absence of pheromone feedback. Then we show that pheromone feedback can lead to a loss of stability of the equilibrium and onset of sustained oscillations in the flow in and out of the nest via a Hopf bifurcation. This analysis sheds light on a potential key mechanism that enables arboreal turtle ants to effectively optimize their trail networks to minimize traveled path lengths and eliminate graph cycles.
%We establish when oscillations do arise in the presence of mixed feedback, numerically characterizing a Hopf bifurcation. 
\end{abstract}

%\begin{IEEEkeywords}
%\end{IEEEkeywords}

\section{Introduction}

%\begin{itemize}
%    \item Key point: mixed-feedback/excitable systems are omnipresent; we will show  a potential role of fixed feedback in ant foraging.
%    \medskip 

Oscillatory dynamics are a characteristic feature of computations in complex systems across nature and technology.
For example, oscillations in the flow of nutrients and signaling molecules have been linked to the abilities of the cellular slime mold \textit{Physarum polycephalum} to perform a variety of cognitive tasks, including shortest path computation \cite{boussard2021adaptive,dussutour2024flow,gyllingberg2024minimal}.
Analogously, oscillations in nutrient and signaling molecule flow provide an active mechanism for regulating transport direction and coordinating resource allocation in fungal networks \cite{schmieder2019bidirectional}.
In distributed networks of neurons in the brain, oscillations are also fundamental for a wide range of functions including sensory and cognitive processing, memory, and the integration of neural activity \cite{bacsar2000brain}. 
In engineered systems, oscillatory or ``spiking'' control signals are being embraced as a design principle for generating flexible behaviors across neuromorphic hardware and robotics \cite{sepulchre2019control,sepulchre2018excitable,sepulchre2022spiking,cathcart2024spiking,JuarezAlvarez2025}.

In this paper we study the onset of sustained oscillations in the flow into and out of the nest during foraging for the arboreal turtle ant \textit{Cephalotes goniodontus} (henceforth referred to as turtle ant), in order to gain mechanistic insight into factors that enable turtle ant trail network optimization. 
To do this, we propose and analyze a compartmental model of turtle ant foraging dynamics. %, in order to gain scientific insight into its trail network adaptation. 
Turtle ants are known to maintain complex trail networks along tree branches and leaves, connecting their nests and food sources \cite{gordon2012dynamics,gordon2017local}. 
In recent work, it was shown that turtle ants can change these trail networks over time to minimize total traveled path length and to eliminate redundant graph features such as cycles \cite{chandrasekhar2021better,garg2023distributed}.
However, doing this requires the rate of flow of ants along trails to increase over time, which cannot be sustained indefinitely due to finite ant volume. 
%In recent work, an algorithmic model of this foraging process showed that 
%A distributed algorithm model of this optimization requires the flow rate of ants through the network to increase for convergence to the shortest path, and the ants have been observed increasing their flow rate in response to food \cite{garg2023distributed}. 
Inspired by these findings, we will explore conditions under which sustained oscillations in the flow rate of turtle ants emerge from interactions between different groups of ants at the nest, and from feedback loops between the ants' return rate to the nest and the pheromone along their trails. 
Under oscillatory conditions, regular periods of flow rate growth that enable trail network optimization are followed by periods of flow rate decay. 
We conjecture that such oscillations may be a key feature of the distributed computation performed by turtle ants during foraging.% as they enable sustained decentralized learning. 
%, with periods of growth followed by periods of decay. Such conditions can be a key factor 

%Since the total colony population is constrained to a finite volume, the flow rate cannot increase indefinitely. 
%an increase in flow rate suggests there may be a later decrease in flow rate, especially in the case of one nest not interacting with other nests. 
%Therefore, we conjecture that oscillations in the flow rate may play a key role in trail network optimization. 

Mixed feedback is a key feature of oscillatory behavior \cite{sepulchre2019control,sepulchre2018excitable,JuarezAlvarez2025}. 
%It is a mechanism that creates and maintains oscillations and excitations, as 
In a mixed-feedback system, a positive feedback loop drives the state of a system %for systems 
away from equilibrium, while a negative feedback loop prevents these excitations from becoming unbounded. 
In this paper we introduce a dynamic model of turtle ant foraging in which there is mixed feedback between the ants' flow between different behavioral compartments and the amount of pheromone along their trails. \textcolor{black}{We will show that this mixed feedback provides a pathway to sustained oscillations in the flow rate of ants entering and exiting a nest.}
%We will show that this mixed feedback is  key for explaining the emergence of sustained oscillations in the flow rate of ants entering and exiting a nest.

Previous work \cite{chandrasekhar2021better,garg2023distributed} explored how the rate of pheromone deposition and the rate of pheromone decay allow the ants to maintain and repair a trail, showing that an increase in flow was necessary for smoothing out graph cycles. In this work, we build on these insights and explore a mechanistic explanation for \textit{how} periodic increases in the flow might emerge. To do this we introduce the effect of unloading at the nest, which was not considered in previous work. Turtle ants forage for nectar and must unload it by trophallaxis (i.e. mouth-to-mouth transfer) to another ant at the nest. Our modeling incorporates the joint effects of unloading nectar and of pheromone feedback. We will show that together, these two effects can give rise to oscillations.

In our model, we differentiate between \textit{trail ants} whose job it is to look for nectar and bring it back to the nest, and \textit{nest ants} whose job it is to receive nectar from the returning foragers and bring it inside.
We make a simplifying assumption that the foraging ants do not permanently return to the nest, and that nest ants to do not go out to forage. 
To model interactions between trail ants and those that remain at the nest, we utilize a compartmental approach \cite{sumpter2003modelling}, also common in the modeling of epidemics and other spreading processes \cite{brauer2008compartmental,gracy2025modeling}. 
In our model, a trail ant supplying food must come into contact with a nest ant receiving food before it can return to the trails to forage.
The competition between pheromone deposition by ants along their trail network and depletion of pheromone via evaporation is then a source of mixed feedback.
Deposition %of pheromone by ants that move around the trails 
is a source of positive feedback as it indirectly increases the rate of return of ants to the nest. 
%If more pheromone is present along a trail segment, the spatial flow of ants is likely to concentrate towards that segment, which in turn increases the rate of pheromone deposition.  %ants are more likely to return back to the nest in a timely manner, %more ants travel along these trails, 
%depositing even more pheromone on their way back. %, causing positive feedback. 
On the other hand, %pheromone 
evaporation is a source of negative feedback, since it depletes the pheromone and slows the rate of return.
% and a decrease in the amount of pheromone on the trails can cause the foraging ants to spread out more across the trail network and return back to the nest at a lower rate. %get lost, i.e. decrease the ants' rate of return to the nest
%, which will in turn lower how much pheromone is being deposited at any particular location.
%We model the rate of return of ants from the trail to the nest as proportional to the overall pheromone concentration.
We will show that together, the effects of trophallaxis and of mixed pheromone feedback can destabilize the flow of ants %into and out of the nest 
and cause it to oscillate. %This work will provide novel insights into the role of oscillations in distributed cognition 

%This mechanism can provide key insights into real world systems, such as the Hodgkin–Huxley model of neurons, Josephson junctions, and two wheeled trailers (CITATIONS). We will show a potential role of mixed feedback in arboreal turtle ant foraging.

%    \item Background on turtle ants, foraging; mention shortest path computation; mention that the algorithm needs a nonconstant flow rate; possible mechanism = oscillations, would be consistent with slime mold observations + role of oscillations there .
%    \medskip 

%  Slime mold \textit{Physarum polycephalum} has also been shown to solve shortest-path problems, and it has been shown that oscillations likely play an important role in its ability to perform cognitive tasks \cite{boussard2021adaptive,gyllingberg2024minimal}. This further indicates that oscillations may play a role in shortest path planning. 

% talk about how ants really do change their flow rate?

%\item What this paper is about: we propose and analyze a simple compartmental model of turtle ant foraging to explore possible mechanism that could lead to oscillations and therefore enable the group to effectively optimize paths along trails; we show that mixed environmental feedback plays a key role 
%\medskip

%In this letter we will introduce and analyze a compartmental model for turtle ant foraging, following an approach similar to \cite{sumpter2003modelling} with added environmental feedback and 

%\item paragraph enumerating contributions explicitly. 

%\medskip

The following are the contributions of this paper. First, we introduce a new compartmental model that describes turtle ant foraging. Our model captures interactions between different groups of ants as they exchange food at the nest as well as a mixed feedback loop between the foraging ants' rate of return to the nest and the dynamics of the pheromone along the trail. We prove that the model is well-posed.
Second, we prove that in the absence of pheromone feedback the foraging model is globally asymptotically stable over its domain and find an explicit solution for the ant volumes in different behavioral compartments at the nest and on the trail at equilibrium.  
%This result rules out the possibility of oscillatory flow rates from purely trophollaxis interactions. 
Third, we study the foraging model with pheromone feedback and derive implicit conditions for existence and uniqueness of an equilibrium solution. Numerically, we illustrate that this equilibrium can lose stability in a Hopf bifurcation at which is an onset of sustained oscillations in the model, which we show in simulation. Together, these results highlight the potential role of trophallaxis interactions, pheromone deposition, and mixed feedback in the optimization of turtle ant trail networks. 

This paper is structured as follows. In Section \ref{sec:prelims} we state preliminaries. In Section \ref{sec:model} we motivate and introduce our foraging model. In Section \ref{sec:analysis} we analyze the model with and without pheromone feedback, and present numerical studies. In Section \ref{sec:discussion} we conclude and discuss future work.
%introduction and analysis of a compartmental model with and without mixed environmental feedback modeling foraging behavior of arboreal turtle ants. We prove that without mixed feedback from the environment, the model always settles into equilibrium. We show that the presence of mixed feedback from the environment can disrupt this equilibrium and induce oscillations. We characterize this Hopf bifurcation and (RESULTS). % not sure how this is different than the abstract other than being a little bit more concrete? maybe number things explicitly

%\end{itemize}

%Connect to mixed feedback control literature \cite{sepulchre2019control} \cite{sepulchre2018excitable}

%Turtle ant foraging key papers \cite{chandrasekhar2021better} \cite{garg2023distributed} 

%Slime mold oscillations papers \cite{boussard2021adaptive} \cite{gyllingberg2024minimal}

%Compartmental models of social insect foraging \cite{sumpter2003modelling}

%Related work: coupled SIS bi-virus modeling \cite{gracy2025modeling}

\section{Mathematical Preliminaries \label{sec:prelims}}

%\subsection{Definitions}
Let $n\in \mathbb{N}$. The closure of $S\subset \mathbb{R}^n$ is $\bar S=\cap_{V\supset S, V\text{closed}}V$. The interior of $S$ is $S^\circ =\cup _{U\subset S,U\text{open}}U$. The boundary of $S$ is $\partial S = \bar S \setminus S^\circ$.
%The spectrum of $A\in M_{n\times n}(\mathbb{R})$, denoted $\text{spec}(A)$, is the set of eigenvalues of $A$. 
The spectral abscissa of $A\in M_{n\times n}(\mathbb{R})$ is $\sigma(A) = \max_{\lambda \in \text{spec}(A)}\{\text{Re}(\lambda)\}$.
%\textcolor{black}{The spectral abscissa of $A\in M_{n\times n}(\mathbb{R})$ is $\sigma(A)= \max_{\lambda \in \text{spec}(A)}\{\text{Re}(\lambda)\}$.}
A system of differential equations with continuous first derivatives on $\mathbb{R}^n$ given by $\dot x_i = f_i(x_1,...,x_n) = f_i(x), i\in \{1,..,n\}$, or equivalently $\dot x = f(x)$, with $f:\Gamma \mapsto \R^n, \Gamma \subseteq \R^n$ is called \textit{competitive} if $\frac{\partial f_i}{\partial x_j}\le  0 \forall i\neq j \forall x \in \Gamma$. \textcolor{black}{Next, we restate Proposition II.1 %, its definition and proof following 
from \cite[Theorem 2.3]{doi:10.1137/0513013}}.

%\subsection{Propositions}

\begin{proposition}[Flow of planar competitive systems] \label{prop:comp_sys}
    Suppose $x \in \R^2$ and a system given by $\dot x = f(x)$ is competitive. Suppose $\Gamma = \R^2 \text{ or } \R_+^2$ and let $y:[0,\infty)\mapsto \Gamma$ be the solution through $y(0)$. Then either $\|y(t)\| \to \infty$ or $y(t)$ converges to some point of $\bar \Gamma$ as $t\to \infty$.% Furthermore, $y_1(t), y_2(t)$ are both monotone. 
\end{proposition}

%\begin{proposition}[Hopf theorem]\cite[Theorem 3.4.2]{guckenheimer2013nonlinear} \label{prop:Hopf}
%    Suppose that a system with parameter $\mu$ given by $\dot x = f(x,\mu), x \in \R^n, \mu\in \R$ has an equilibrium $f(x_0^*,\mu_0)=0$. Suppose the jacobian of $f(x_0^*,\mu_0)$, denoted $J(x_0^*,\mu_0)$ has a pair of complex conjugate eigenvalues $\lambda, \bar \lambda$ and no other eigenvalues with zero real part. Then there is a smooth curve of equilibria $(x(\mu),\mu)$ with $x(\mu_0)=x_0^*$ and $\lambda, \bar \lambda$ vary smoothly with $\mu$. If $\frac{d}{d\mu^*}\text{Re}( \lambda(\mu))|_{\mu = \mu_0} \ne 0$ then there is a unique three-dimensional center manifold passing through $(x_0^*,\mu)$ in $\mathbb{R}^n\times \mathbb{R}$ and a smooth system of coordinates (preserving the planes ($\mu = $const.)) for which the Taylor expansion of degree 3 on the the center manifold is given in polar coordinates by $\dot r = r(d\mu + ar^2), \dot \theta = \omega + c\mu + br^2$.
%    If $a \ne 0$ there is a surface of periodic solutions in the center manifold. If $a<0$, then these periodic solutions are stable limit cycles, if $a>0$, the periodic solutions are repelling. %We call $(x_0^*,\mu_0)$ a Hopf Bifurcation point.
%\end{proposition}

\section{Model Setup \label{sec:model}}

\subsection{Compartmental Foraging Model}

A colony of arboreal turtle ants occupies multiple nests along the trail network. % as it ages and grows. 
In this paper we study the simplest scenario in which there is only a single nest. 
%This scenario corresponds to a young colony.
Consider two interacting groups of ants with $N > 0$ ants that remain at the nest and $M > 0$ ants that explore the trails. 
We assume that there are no transitions between these groups, i.e. that trail ants do not permanently return to the nest and that nest ants do not go out onto the trails. \textcolor{black}{This is a reasonable assumption for time scales within a day %, since it is common for ants to do a particular task for a while and then switch to another 
\cite{gordon_division_2016}.}
Under this assumption, the volumes $N$ and $M$ are constant. 
%We suppose the total number of ants in each group is constant, with $N$ and $M$ ants, respectively. 

Following a compartmental modeling paradigm \cite{sumpter2003modelling}, assign ants within both the nest and trail groups to one of two mutually exclusive behavioral states.  
The trail group is made up of $\bar F$ \textit{forager} ants out on the trails actively exploring for food and $\bar S$ \textit{supplier} ants that have acquired food and are waiting to hand it off to a nest ant. 
By conservation of volume, $\bar F + \bar S = M$. 
Analogously, we split the nest ant group into $\bar R$ \textit{receiver} ants picking up food from suppliers outside of the nest and $\bar I$ \textit{interior} ants that are inside of the nest, with $\bar R +\bar I = N$. Finally, we define normalized variables $R = \bar R/N, I = \bar I/N, F = \bar F/M, S = \bar S/M$ such that $R+I=1$ and $S+F=1$. Observe that $R,S,I,F$ are each constrained to the unit interval $[0,1]$. \textcolor{black}{In taking this modeling approach, we have abstracted away the geometry of the trail network.}%}, in favor of the simplicity of modeling the current tasks of ants.} 

Next, we consider transitions between the behavioral compartments within the nest and trail groups. 
We assume that foragers return to the nest to supply food at constant transition rate $\gamma > 0$, and interior ants become receivers at a transition rate $\alpha > 0$. 
To transition in the other direction, a supplier from the trail group must come in physical contact with a receiver from the nest group in order to hand off food.
Therefore the transition rates from supplier back to forager, as well as those from receiver back to interior ant, will be state-dependent.
By the law of mass action, the transition rate from $\bar R$ to $\bar I$ is $\beta \bar S/N$ with scaling parameter $\beta > 0$.
Analogously, the transition rate from $\bar S$ to $\bar F$ is $\beta \bar R/M$.
Assuming that the volume variables are continuous, observing that volume conservation implies $\dot{\bar R} + \dot{\bar I} = 0$, $\dot{\bar S} + \dot{\bar F} = 0$, and applying normalization, we arrive at a two-dimensional dynamical system that describes the dynamics of the nest and trail variables over time as
%$\beta \frac{N}{M}RS$ respectively, with parameter $\beta > 0$. Since $R+I = 1$ and $F+S=1$ the dynamics are two dimensional:
\begin{subequations} 
\label{eq:model_2d}
    \begin{align}
        \dot R &= \alpha(1 - R) - \beta \frac{M}{N} R S, \label{eq:Rdot}\\
        \dot S &= \gamma (1 - S) - \beta \frac{N}{M} R S. \label{eq:Sdot}
    \end{align}
\end{subequations}
The variables $I(t)$ and $F(t)$ can be recovered from volume conservation as $I(t) = 1 - R(t)$ and $F(t) = 1 - S(t)$. Key features of \eqref{eq:model_2d} are summarized in Fig. \ref{fig:nest-compartments}.

%\begin{figure}%[h]
  %  \centering
  %  \includegraphics[width=0.75\linewidth]{ants.png}
 %   \caption{Graphical summary of compartmental foraging model  \eqref{eq:model_2d}.}
 %   \label{fig:nest-compartments}
%\end{figure}

\begin{figure}%[b]
    \centering
    \includegraphics[width=.78
    \linewidth]{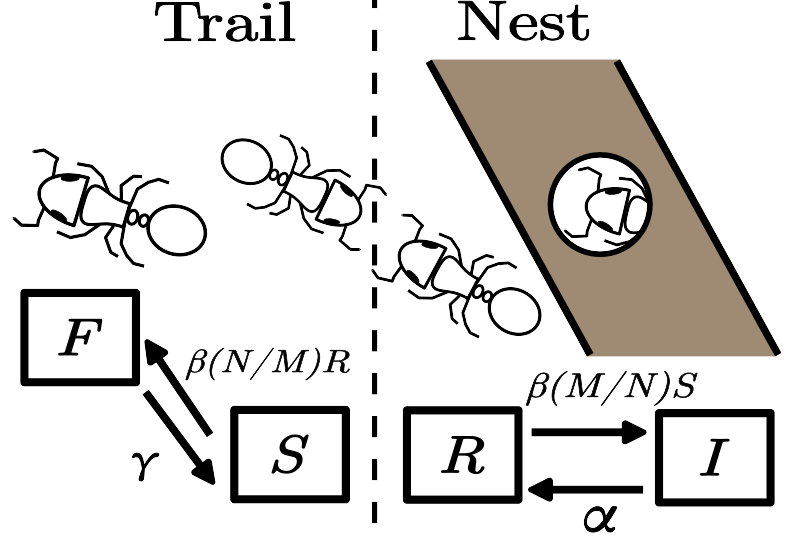}
    \caption{Graphical summary of compartmental foraging model  \eqref{eq:model_2d}.}
    \label{fig:nest-compartments}
\end{figure}

\subsection{Pheromone Feedback}

The foraging model \eqref{eq:model_2d} accounts only for effects of direct interaction between ants.
However, the ants' foraging dynamics are also coupled with and regulated by the ants' pheromone deposition along trails. To model this coupling, we introduce a new dynamic variable $p \ge 0$, \textcolor{black}{representing an amount of pheromone that is deposited by the flow of ants to and from the nest. }% Since we do not }  
As ants move along the trail network, they deposit pheromone which increases the amount $p$. The pheromone also evaporates at fixed rate $\mu > 0$, which decreases the \textcolor{black}{amount} $p$. The rate of pheromone deposition is proportional both to the rate of return of foragers to supply food at the nest, captured by $\gamma \bar F = \gamma M (1 - S)$, and to the rate of return of suppliers to the foraging state, captured by $\beta \frac{\bar R}{M} \bar{S} = \beta N R S$.
An increase in either of these rates will increase the pheromone amount $p$, since movement in either direction (towards or away from the nest) results in new pheromone deposition.
We let $\nu > 0$ be a proportionality constant that tunes the deposition rate. Putting these pieces together, the temporal dynamics of the amount of pheromone $p$ is governed by
\begin{equation}
    \label{eq:p}
        \dot p = - \mu p + \nu \big( \gamma  M (1 - S) + \beta N R S \big).
\end{equation}

%\begin{figure}
  %  \centering    \includegraphics[width=0.8\linewidth]{hill.png}
  % \caption{Shape of the saturating Hill function $f(p) = \frac{k}{1+ (p_0/p)^n}$ modeling the relationship between pheromone level and change in ant return rate for different choices of $n$. The other parameters of $f(p)$, $k$ and $p_0$ are fixed at $k = 1, p_0 = 0.2$. In numerical simulations we used these values for $k$, $p_0$, and chose $n = 4$.}
  %  \label{fig:hill}
%\end{figure}

Next we consider the rate $\gamma$ at which forager ants return to the nest to become suppliers. This rate was assumed to be constant in the baseline model \eqref{eq:model_2d}. In reality, $\gamma$ is closely related to the amount of pheromone $p$. Foragers rely on pheromone to navigate the colony's trail networks. At high amounts, pheromone is easy for foragers to detect, so navigating back to the nest to supply food is easier. At low amounts, navigation becomes more difficult and the return rate of foragers will decrease as they may get lost or take a longer path towards food and back to the nest. To model this pheromone feedback, we take the return rate parameter $\gamma$ to be dynamic, with dynamics

\begin{equation}
    \label{eq:gamma}
        \tau_{\gamma} \dot \gamma = - \gamma +  f(p) + \gamma_0,
\end{equation}
where $f:\mathbb{R}^{\geq 0} \to \mathbb{R}^{\geq 0}$ a monotonically increasing function of $p$, $\tau_\gamma>0$ a time scale, and $\gamma_0\ge0$ a base rate of return to the nest. %Two natural choices for $f$ are a linear function $f_\text{lin}(p) = kp$ with $k\ge0$, and 
In numerical simulations we choose $f$ to be a saturating Hill function $f(p) = \frac{k}{1 + (p_0/p)^n}$, with positive parameters $k, p_0,n$. \textcolor{black}{This choice is motivated by the assumption that beyond a threshold amount of pheromone, more pheromone does not influence the decision about which path to take from a junction in the vegetation.}

%certain point, increasing pheromone concentration no longer continues to increase ease of navigation.}
%For our numerical simulations, we will choose $f_\text{Hill}$ since it saturates. 

Together, our foraging model with pheromone feedback is a four-dimensional dynamical system given by \eqref{eq:model_2d}, \eqref{eq:p}, and \eqref{eq:gamma}. The parameters of this model summarized in Table \ref{table:1}. %\textcolor{black}{and modeling assumptions are summarized in table \ref{table:assumptions}}. 
Note that the coupling between \eqref{eq:p} and \eqref{eq:gamma} is a source of mixed dynamic feedback in the model, as the evaporation of pheromone \textit{inhibits} the rate of return while the deposition of pheromone \textit{reinforces} the rate of return. Our analysis in Section \ref{sec:analysis} will show that this mixed feedback is key to the emergence of limit cycles in the compartment variables, and therefore also in the flow rates of the groups. This hints at a key mechanism behind the ants' success at navigating trail networks, as a time-varying rate of flow along trails has previously been shown to play a key role in trail network optimization that enables ants to minimize traveled path lengths and eliminate cycles \cite{garg2023distributed}.

\begin{table}[h!]
\centering
\begin{tabular}{|c|p{6.5cm}|} 
 \hline
 Parameter & Summary \\ [0.5ex] 
 \hline\hline
 $N$ & Nest ant volume.\\
  $M$ & Trail ant volume.\\
 $\alpha$ &  Transition rate from interior to receiver at the nest.  \\
 $\beta$ & Interaction rate scale between suppliers and receivers.   \\
 $\mu$ & Pheromone evaporation rate. \\
 $\nu$ & Pheromone deposition rate.\\
 $\gamma_0$ & Base rate of $\gamma$, transition rate from forager to supplier; equal to $\gamma$ in the 2d model \eqref{eq:model_2d}.\\
 $\tau_\gamma$ & Characteristic timescale of the response of the return rate $\gamma$ to changes in $p$ \\
 $k$ & Scaling constant of $f(p)$ \\
 $p_0$ & Saturation midpoint parameter of $f(p)$ \\
 $n$ & Slope parameter of $f(p)$\\
 \hline
\end{tabular}
\caption{Summary of model parameters in \eqref{eq:model_2d},\eqref{eq:p},\eqref{eq:gamma}.}
\label{table:1}
\end{table}

%\textcolor{black}{
%Modeling Assumptions and Idealizations:
%\begin{itemize}
 %   \item There are no transitions between ants on the trail and ants that remain at the nest.
 %   \item Ants deposit a single pheromone.
  %  \item Ants that search for food are well mixed on the trails and return with constant rate $\gamma$.
   % \item We ignore the geometry of the trail network, considering only transitions between compartments and an abstraction of pheromone concentration.
    %\item $f(p)$ saturates, i.e. there is a limit to how much pheromone can help with navigation.
%\end{itemize}}

%\begin{table}[h!]
%    \centering
%    \begin{tabular}{p{8.5cm}}
%    \hline
%         \textcolor{black}{There are no transitions between ants on the trail and ants that remain at the nest.}   \\ \hline \textcolor{black}{We ignore the geometry of the trail network, considering only transitions between compartments and coupled pheromone dynamics.} \\ \hline \textcolor{black}{$f(p)$ saturates, i.e. there is a limit to how much pheromone can help with navigation.} \\ \hline \textcolor{black}{ The ant colony has infinite population, state variables are continuous.} \\ \hline \textcolor{black}{Ants deposit a single pheromone.}\\ \hline
         %Ants that search for food are well mixed on the trails and return to supply with transition rate $\gamma$.\\ \hline \end{tabular} \caption{\textcolor{black}{Modeling Assumptions and Idealizations}} \label{table:assumptions}
%\end{table}

\subsection{Well-Posedness}

In the following Theorem we establish the well-posedness of the two-dimensional foraging model \eqref{eq:model_2d} and of its four-dimensional extension with pheromone feedback \eqref{eq:p},\eqref{eq:gamma} by proving that the variables $R,S,p,\gamma$ do not take on values outside of the interpretable range under the flow of the dynamical system.

\begin{theorem}\label{thm:well-posedness} The following statements hold.

    1) The compact set $\Omega_1=[0,1]\times [0,1]$ is positively invariant under the flow of \eqref{eq:model_2d} with static parameter $\gamma > 0$;
    
    2) The set $\Omega_2=[0,1] \times [0,1] \times \mathbb{R}^{\geq 0} \times \mathbb{R}^{\geq0} $ is positively invariant under the flow of \eqref{eq:model_2d},\eqref{eq:p},\eqref{eq:gamma}. 
\end{theorem}
\begin{proof} %In both cases we need to apply Nagumo's theorem \cite[Theorem 4.7]{blanchini2008set} 

    1) Let $\partial \Omega_1$ be the boundary of $\Omega_1$, which includes the sides and corners of the unit cube. \textcolor{black}{Table \ref{table:III1proof} breaks $\partial \Omega_1$ into cases, showing in each case that the flow along $\partial \Omega_1$ is into $\Omega_1$.
    %is $\partial \Omega_1 = (0,1)\times \{0\}\cup (0,1)\times\{1\}\cup \{0\}\times (0,1)\cup \{1\}\times (0,1)\cup \begin{bmatrix}0\\0\end{bmatrix}\cup\begin{bmatrix}1\\1\end{bmatrix}\cup \begin{bmatrix}0\\1\end{bmatrix}\cup \begin{bmatrix}1\\0\end{bmatrix}$. 
    %If $R\in(0,1),S =0$, then $\dot S = \gamma >0$. If $R \in (0,1), S = 1$, then $ \dot S = -\beta \frac{N}{M}R<0$. If $R=0,S\in (0,1)$, then $\dot R = \alpha >0$. If $R=1, S\in (0,1)$, then $\dot R = -\beta\frac{M}{N}S<0$. If $R=S=0$, then $\dot R = \alpha >0, \dot S = \gamma >0$. If $R = S = 1$, then $\dot R = -\beta\frac{M}{N}<0, \dot S = - \beta\frac{N}{M}<0$. If $R=0,S=1$, then $\dot R = \alpha >0, \dot S = 0$. If $R = 1, S=0$, then $\dot R = 0, \dot S = \gamma >0$. 
     With all the cases checked, Nagumo's theorem \cite[Theorem 4.7]{blanchini2008set} implies $\Omega_1$ is positively invariant.}

\begin{table}[h!]
\centering
\begin{tabular}{|c|p{3.0cm}|}
 \hline
 Case & Implication \\ [0.5ex] 
 \hline\hline
 $S=0$, $R\in [0,1]$ & $\dot S = \gamma >0$ \\ \hline
 $S=1$, $R\in (0,1]$ & $\dot S = -\beta\frac{N}{M}R <0$ \\ \hline
 $R=0$, $S\in (0,1)$ & $\dot R = \alpha >0$ \\ \hline
$R = 1, S\in (0,1)$ & $\dot R = -\beta \frac{M}{N}S<0$ \\ \hline
$S=1, R =0$ & $\dot R = \alpha>0$ \\ \hline
\end{tabular}
\caption{\textcolor{black}{Proof the flow of \eqref{eq:model_2d} is into $\Omega_1$ along its boundary. Note that the first four cases together also check three corners of $\partial \Omega_1$.}}
\label{table:III1proof}
\end{table}

    2) We check that the flow along boundaries of $\Omega_2$ does not exit the set. \textcolor{black}{If $p=0$, then $\dot p  = \nu \big( \gamma  M (1 - S) + \beta N R S \big) \ge 0$, and if $\gamma = 0$, then $\dot \gamma = \frac{1}{\tau_\gamma}(f(p)+\gamma_0) >0$. For any other condition, the previous part implies the desired result.}% Other cases are consistent with 1).}
    %If $R\in [0,1], S = 0, p,\gamma \in [0,\infty)$, then $\dot S = \gamma \ge 0$. If $R,S \in [0,1], p = 0, \gamma \in [0,\infty)$, then $\dot p = \nu(\gamma M(1-S)+\beta N RS)\ge 0$. If $R,S \in [0,1], p \in [0,\infty), \gamma = 0$, then $ \tau_\gamma\dot\gamma = f(p)+\gamma_0 >0$. Other cases are consistent with 1). %Theorem statement follows from Nagumo's theorem.
\end{proof}

\section{Analysis\label{sec:analysis}}

\begin{figure}
    \centering
    \includegraphics[width=0.9
    \linewidth]{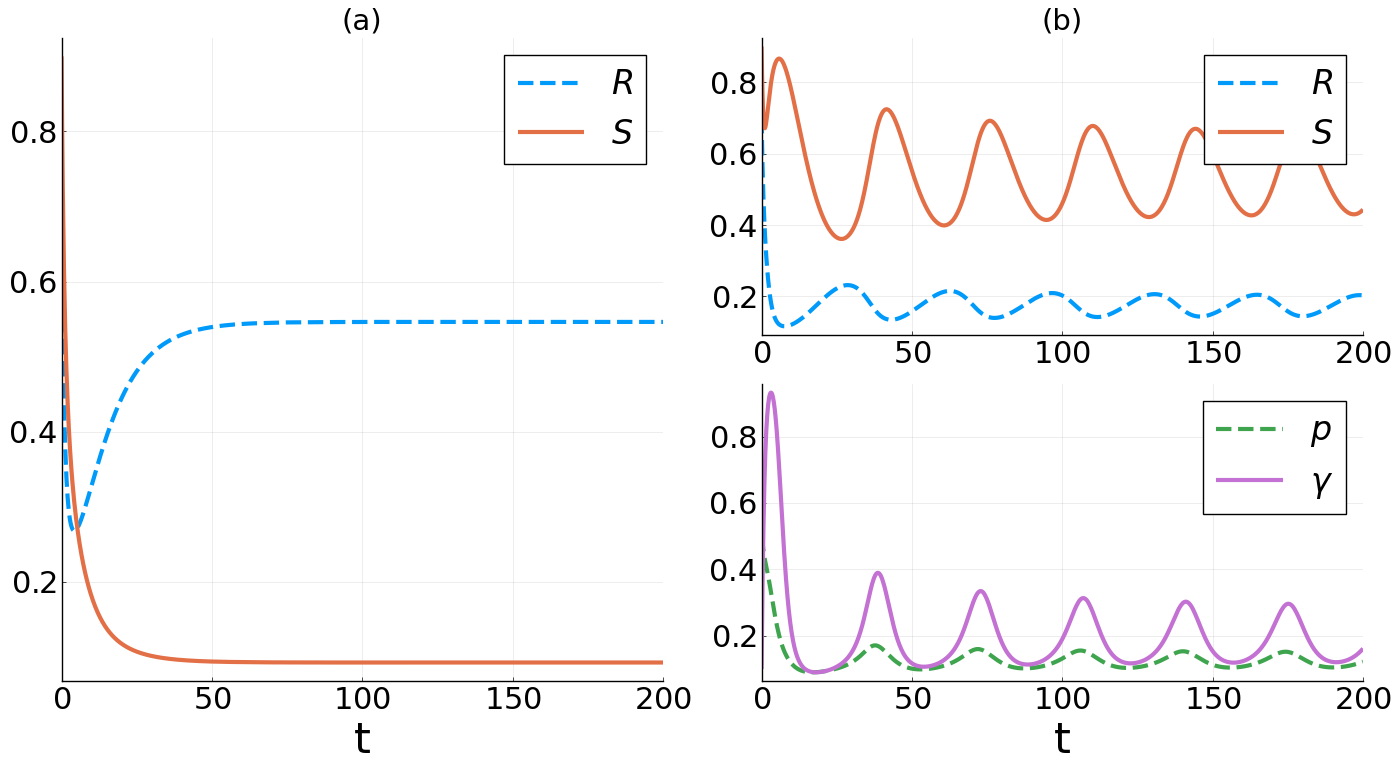}
    \caption{(a) Numerically computed trajectory of \eqref{eq:model_2d} with parameters $N=M=1, \alpha = 0.1, \beta=0.9,\gamma=0.05$ settling into the globally attracting equilibrium. (b) Numerically computed trajectory of \eqref{eq:model_2d},\eqref{eq:p},\eqref{eq:gamma} With the same parameters as in (a) and additionally $\mu = 0.6,\nu = 0.45,\tau_\gamma = 1,\gamma_0 = 0.05, k = 1, p_0 = 0.2,n=4$. The initial condition of (a),(b) is $R(0)=0.7$, $S(0)=0.9$; 
    %, $R$ and $S$ initial conditions in (a) and (b) are identical and 
    in (b), additionally $p(0) = 0.5$ and $\gamma(0)= 0.1$.}
    \label{fig:trajectories}
\end{figure}

In this Section we study the foraging model \eqref{eq:model_2d},\eqref{eq:p},\eqref{eq:gamma}, with the goal of classifying conditions under which oscillations in the compartment variables, and therefore in their rates of change, are possible.

\subsection{Model without pheromone feedback\label{sec:2dmodel}}

First we consider the compartmental foraging model \eqref{eq:model_2d} with a static interaction rate parameter $\gamma$. In the following Theorem we prove that under such static parameter assumptions, the model is globally asymptotically stable, i.e. oscillations are impossible.

\begin{theorem}[Global Asymptotic Stability] \label{thm:2D-EQ}
    Consider \eqref{eq:model_2d} over $\Omega_1 = [0,1] \times [0,1]$. 
    
    1) The model has a unique fixed point $(R,S) = (R^*,S^*)$ in $\Omega_1$, given by $S^* = \frac{1}{2A}(A-B-1 + \sqrt{(A-B-1)^2 + 4AB}), R^* = B(1-S^*)/S^*$ where $A = \gamma M^2 / \alpha N^2$ and  $B = \gamma M/ \beta N$ 

    2)  $(R^*,S^*)$ is globally attracting on $\Omega_1$.
\end{theorem}
\begin{proof}
1) At equilibrium, $\dot R = \dot S = 0$, which implies $\beta R^* S^* = \alpha\frac{N}{M}(1-R^*) = \gamma \frac{M}{N}(1-S^*)$. Solving this expression for $R^*$ yields $R^* = 1-A(1-S^*)$, which can then be plugged into \eqref{eq:Sdot} to arrive at a quadratic equation of the form 
%. Plugging that into $\dot S^* =0$ gives $0 = \gamma(1-S^*)-\beta\frac{N}{M}S^*(1-A(1-S^*))$. Multiplying by $-\frac{M}{\beta N}$ and collecting like terms leaves us with 
%the quadratic 
$0 = A(S^*)^2 +(1+B-A)S^* - B$ which has two solutions $S^*_\pm = \frac{1}{2A}\left(A - B - 1 \pm \sqrt{(A-B-1)^2 +4AB}\right)$. 
These solutions generate $R^*_\pm = B\frac{1-S^*_\pm}{S^*_\pm}$. 
%We will refer to $R^*_+, S^*_+$ as $R^*, S^*$

For compactness let $C_1 =A-B-1$ and $C_2 = \sqrt{(A-B-1)^2 +4AB}$. Next, we will show that $R^*_-, S^*_- \notin \Omega_1$. Observe that $C_2 > C_1$ for any $C_1 \in \mathbb{R}$ since $4AB > 0$. Then $C_1-C_2 < 0$ which implies $S_-^* < 0$. % by considering several cases. \textbf{Case 1:} $C_1\leq0$. Then $C_1 - C_2 < 0$ and $S^*_- < 0$.   \textbf{Case 2:} $C_1 > 0$. Observe that $C_2 > C_1$, which means $C_1-C_2 <0$ and therefore $S^*_- < 0$.

Finally, we will show that $(R_+^*,S_+^*) \in \Omega_1$. An analogous argument to above shows that $S_+^* > 0$.  
%Next, we will show that $R^*_-, S^*_- \notin \Omega_1$ and $R^*_+, S^*_+ \in\Omega_1$. If $(A-B-1)<0$, then $S^*<0$ since $(A-B-1)^2 = (1+B-A)^2$ implies both terms of $S^*_-$ are negative. If $A-B-1=0$ then $S^*_ = -\sqrt{4AB}<0$, and if $(A-B-1)>0$, $S^*_- <0$ because $A - B - 1 < \sqrt{(A-B-1)^2 + 4AB}$. $R^*_- = B\frac{1-S^*_-}{S^*_-}<0$ since $S^*_-<0$. A similar argument shows $S^*>0$. 

Suppose $S_+^*<1$. Then % $\frac{1}{2A}(A-B-1 + \sqrt{(1+B-A)^2 + 4AB})<1$. Rearranging this and squaring both sides gives 
$(1+B - A)^2 + 4AB < (A+B+1)^2$. Expanding and canceling like terms results in the condition $0<2A$ which is always satisfied, so $S_+^*<1$. $R_+^* = B\frac{1-S_+^*}{S_+^*}$ implies $R_+^*>0$ since $S_+^* \in (0,1)$. Finally, suppose $R_+^*<1$. Writing this in terms of $S_+^*$ gives $B\frac{1-S^*}{S^*}<1$ which rearranges to $S_+^*> \frac{B}{1+B}$. Substituting our expression for $S_+^*$, multiplying by $2A$, and rearranging yields $\sqrt{(A+B-A)^2 +4AB}>\frac{2AB}{1+B} + 1+B-A$, %. Squaring both sides, dividing by $4AB$ and multiplying by $(1+B)^2$ leaves $(1+B)^2> AB + (1+B-A)(1+B)$ 
which simplifies to $A>0$ which is always true, so $R_+^*<1$. The equilibrium $(R^*,S^*) = (R_+^*,S_+^*)$ is therefore unique in $\Omega_1$.

%working on this part
%To see $R^*,S^*\in [0,1]$, first note that similar arguments to the previous part of this proof shows that $S_+^*, R_+^* \ge 0$. Writing the conditions $S_+^* \le1, R_+^*\le1$ and expanding algebraically gives statements that are trivially always true, so the inequalities themselves must be true as well. So $R^*,S^*\in [0,1]$.

2) The Jacobian of \textcolor{black}{\eqref{eq:model_2d}}
 is: 
 \[J(R,S) = 
 \begin{bmatrix}
 - \alpha - \beta \frac{M}{N}S & - \beta \frac{M}{N}R \\
 - \beta \frac{N}{M}S & - \gamma - \beta \frac{N}{M}R
 \end{bmatrix}
 \]
 Since $S,R \in [0,1]$, all entries of $J$ are non-positive, and the system given by \eqref{eq:model_2d} is \emph{competitive}. %The dynamics are 2-dimensional, and 
 By Proposition \ref{prop:comp_sys}, %theorem 2.3 of \cite{doi:10.1137/0513013} implies
 any trajectory $(R(t),S(t))$ starting from initial condition $(R(0),S(0)) \in \Omega_1$ will either approach an equilibrium or $ \|(R(t),S(t)) \| \to \infty$ as $t \to \infty$. By Theorem \ref{thm:well-posedness}, the bounded region $\Omega_1$ is invariant under the flow of \eqref{eq:model_2d}, %which means 
 \textcolor{black}{so} all trajectories must settle to an equilibrium. Since the equilibrium $(R^*,S^*)$ defined in part 1) is unique in $\Omega_1$, we conclude it must be globally asymptotically stable.   %$x(t)$ must approach $\begin{bmatrix}
 %    R^*\\S^*
 %\end{bmatrix}$.
\end{proof}
A numerical simulation of the foraging model \eqref{eq:model_2d} settling to its unique globally asymptotically stable equilibrium described in Theorem \ref{thm:2D-EQ} is shown in Fig. \ref{fig:trajectories}(a).

\subsection{Full model }
Theorem \ref{thm:2D-EQ} implies that physical interactions between ants during trophallaxis are not sufficient to cause oscillations in the model. An additional mechanism is necessary in order to destabilize the equilibrium and cause such oscillations. In this Section we explore the pheromone feedback \eqref{eq:gamma}, \eqref{eq:p} as a possible mechanism that could drive this loss of stability. First, we establish conditions for existence and uniqueness of an equilibrium in the full model. Next, we numerically study the stability of this equilibrium, focusing on the effects of variations in transition rate parameters $\alpha$ and $\beta$. We will show that for a range of these parameters, the equilibrium becomes unstable. At the onset of this instability, the model undergoes a %\textcolor{black}{(supercritical)} 
Hopf bifurcation, and oscillations emerge.

First, in the following Theorem, we derive conditions for existence and uniqueness of an equilibrium for the full foraging model with dynamic pheromone feedback.

\begin{theorem}[Equilibrium Existence and Uniqueness] \label{thm:4D-EQ}
    Equilibria \textcolor{black}{$(R,S,p,\gamma) = (R^*, S^*, p^*, \gamma^*)$} of the model \eqref{eq:model_2d}, \eqref{eq:p}, \eqref{eq:gamma} are implicitly defined by $S^* = \frac{\alpha}{\beta}\frac{N}{M}\frac{(1-R^*)}{R^*}$, $p^* = 2\frac{\nu}{\mu}\alpha  \frac{N^2}{M}(1-R^*)$, $\gamma^*=f\left(2\frac{\nu}{\mu}\alpha  \frac{N^2}{M}(1-R^*)\right) +\gamma_0= \alpha \beta \frac{N^2}{M}\frac{R^*(1-R^*)}{\beta M R^* - \alpha N(1-R^*)}$, where $R^*$ is a root of this implicit relationship in the unit interval $[0,1]$. At least one equilibrium exists in $\Omega_2$ for all parameter values. Furthermore, any $R^* \in \left(\frac{\alpha N}{\alpha N + \beta M}, 1\right)$.
    %Furthermore, if for all $x \in \left(\frac{\alpha N}{\alpha N + \beta M}, 1\right)$ 
    Let $R_1$ be the first value of $x \in \left(\frac{\alpha N}{\alpha N + \beta M}, 1\right)$ satisfying the equilibrium conditions. Then if $
       \int_{R_1}^1  f'\left(2\frac{\nu}{\mu}\alpha  \frac{N^2}{M}(1-x)\right) - \frac{1}{2} \frac{\nu}{\mu} \beta \frac{\alpha N (1-x)^2 + \beta M x^2}{\left( \alpha N (1-x) + \beta M x\right)^2} dx > 0$, 
    then the implicit relationship above defines a unique equilibrium $(S^*, R^*, p^*, \gamma^*)$ in $\Omega_2$ where $R^* = R_1$.
\end{theorem} 
\begin{proof}
    At an equilibrium, the right hand side of \eqref{eq:Rdot}, \eqref{eq:Sdot}, \eqref{eq:p}, \eqref{eq:gamma} is equal to zero. We will use these relationships to find expressions for $S^*, \gamma^*, p^*$ as a function of $R^*$, and an implicit function whose zeros define the values of $R^*$ at steady state.  From \eqref{eq:Rdot} and \eqref{eq:Sdot} we see that $\gamma^*(1-S^*)M = \alpha (1-R^*)\frac{N^2}{M}$. Plugging this in to \eqref{eq:p} to eliminate $S^*$ we obtain $p^* = 2\frac{\nu}{\mu}\alpha  \frac{N^2}{M}(1-R^*)$. Then \eqref{eq:Rdot} implies $S^* = \frac{\alpha}{\beta}\frac{N}{M}\frac{(1-R^*)}{R^*}$. To obtain an implicit equation for $R^*$, we can notice that \eqref{eq:Sdot} gives $\gamma^* = \beta\frac{N}{M}\frac{R^*S^*}{1-S^*} = \frac{\frac{N^2}{M^2}\alpha R^* \left(\frac{1-R^*}{R^*}\right)}{1 - \frac{\alpha N (1-R^*)}{\beta M R^*}} = \frac{N^2}{M}\alpha \beta \frac{R^*(1-R^*)}{\beta M R^* - \alpha N(1-R^*)} := g(R^*)$, while $\eqref{eq:gamma}$ gives $\gamma^* = f(p^*) + \gamma_0 = f(2\frac{\nu}{\mu}\alpha  \frac{N^2}{M}(1-R^*)) +\gamma_0 := h(R^*)$. Setting these two expression for $\gamma^*$  equal to each other gives our implicit equation for $R^*$.

    To prove existence and uniqueness of these solutions, we study the number of intersections between the functions $h(x)$ and $g(x)$ over the unit interval $x \in [0,1]$, since the zeros of $h(R^*) -  g(R^*) = 0$ generate the equilibria of the model. First, observe that $h(0) = \gamma_0$, $h(1) = \frac{k}{1 + p_0^n} + \gamma_0$, $g(0) = 0$, and $g(1) = 0$. The function $h(x)$ is continuous and monotonically decreasing on $[0,1]$, so $\gamma_0 \leq h(x) \leq \frac{k}{1 + p_0^n} + \gamma_0$. The function $g(x)$ has a discontinuity at $x = \frac{\alpha N}{\alpha N + \beta M}$ where its denominator is zero. Everywhere else in the set interior $(0,1)$, $g(x)$ is continuous and differentiable, with derivative 
    \begin{equation*}
        g'(x) = - \alpha \beta  \frac{N^2}{M} \frac{\beta M x^2 + \alpha N (1-x)^2}{ \left( \alpha N (1-x) + \beta M x \right)^2} < 0.
    \end{equation*}
    Since $g(0) = 0$, $\lim_{x \to \left(\frac{\alpha N}{\alpha N + \beta M}\right)^-} g(x)= - \infty$, and $g'(x) < 0$ for all $x \in \left(0,\frac{\alpha N}{\alpha N + \beta M} \right)$, we conclude that over this interval,  $g(x) <  0$ while $h(x) > 0$ and intersections are impossible, which rules out equilibria in this region.

    Next, we consider the interval $x \in \left(\frac{\alpha N}{\alpha N + \beta M},1\right)$. Observe that $\lim_{x \to \left(\frac{\alpha N}{\alpha N + \beta M}\right)^+} g(x)= \infty$. Since $g(1) = 0$, we see that $h(1) - g(1) = \frac{k}{1+p_0^n} + \gamma_0 > 0$. Since $h(x)$ is bounded on the unit interval, $\lim_{x \to \left(\frac{\alpha N}{\alpha N + \beta M}\right)^+} h(x) - g(x) = - \infty$. Then by the intermediate value theorem, there must exist at least one value $x = R^*$ at which $h(R^*) = g(R^*)$. This establishes existence of an equilibrium in this interval. This equilibrium is unique if the difference $h(x) - g(x)$ does not cross zero a second time. The integral condition in the Theorem statement imposes this non-intersection.%, from which uniqueness of the equilibrium follows.
\end{proof}

Theorem \ref{thm:4D-EQ} establishes that the model \eqref{eq:model_2d},\eqref{eq:p},\eqref{eq:gamma} always has at least one equilibrium. Furthermore, under mild assumptions about the shape of the function $f(p)$ that defines the relationship between the amount of pheromone and the ants' return rate to the nest, this equilibrium is unique. Next, we study the stability of this equilibrium.  The Jacobian \textcolor{black}{$J(R^*,S^*,p^*,\gamma^*)$} of the model evaluated at its equilibrium, %\textcolor{black}{henceforth $J$, is a $4\times4$ matrix depending on the equilibrium $R^*,S^*,p^*,\gamma^*$,  
is given by the matrix
\begin{equation}\begin{bsmallmatrix} -\alpha - \beta \frac{M}{N} S^* & - \beta \frac{M}{N} R^* & 0 & 0 \\     - \beta \frac{N}{M} S^* & - \gamma^* - \beta\frac{N}{M} R^* & 0 & 1-S^* \\   \nu \beta N S^* & \nu (- \gamma^*  M + \beta N  R^*) & - \mu & \nu M (1-S^*) \\    0 & 0 & \frac{f'(p^*)}{\tau_{\gamma}} & - \frac{1}{\tau_{\gamma}}\end{bsmallmatrix}\label{eq:jac_4d}\end{equation}
 uniquely defined by a given choice of model parameters via the implicit relationship derived in Theorem \ref{thm:4D-EQ}, provided the stated uniqueness condition is satisfied. When all eigenvalues of \textcolor{black}{\eqref{eq:jac_4d}} have negative real part, the equilibrium is locally exponentially stable. 

\begin{figure}[t]
    \centering
    \includegraphics[width=0.90\linewidth]{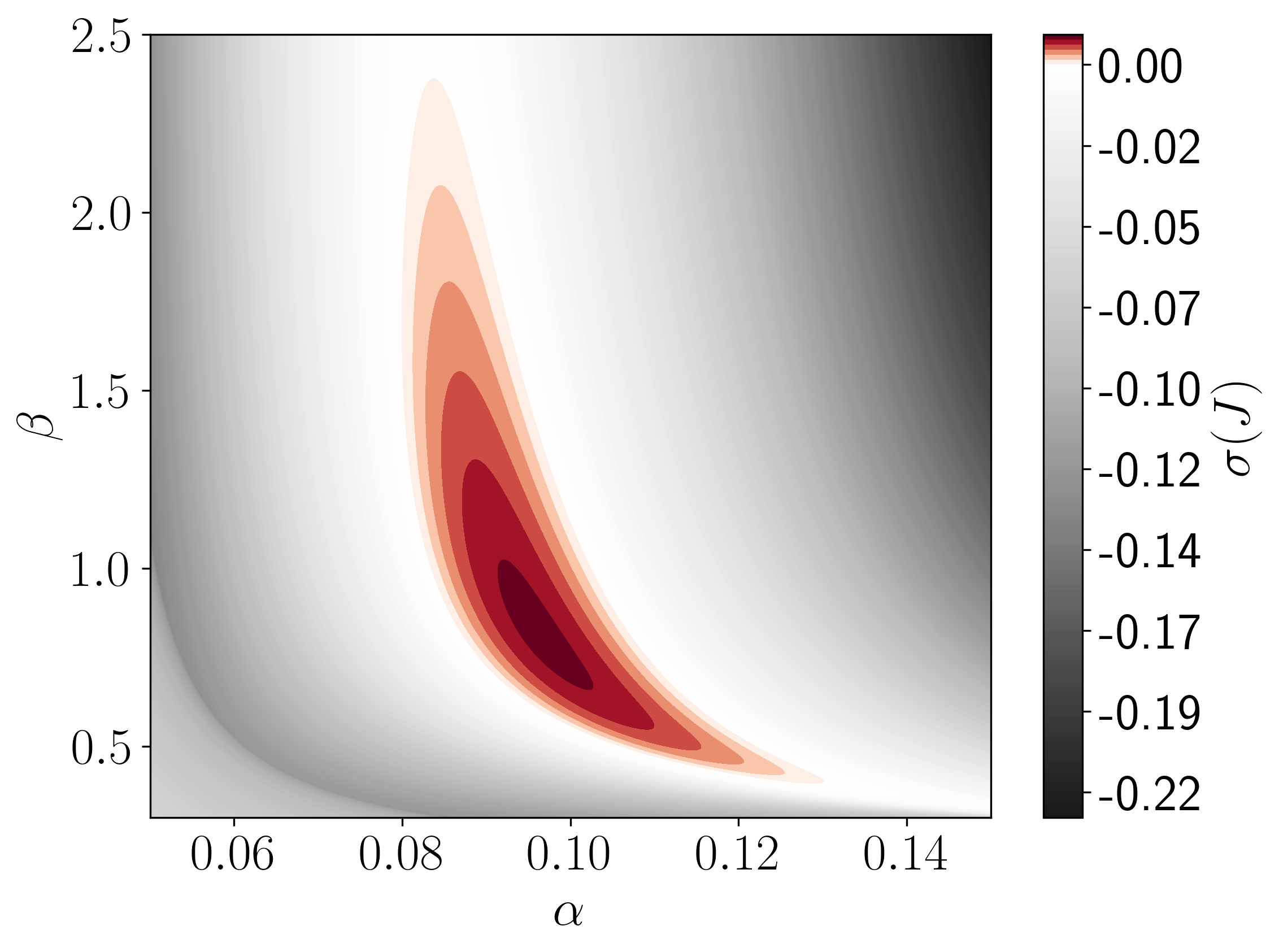}
    \caption{Contour plot of the spectral abscissa of the Jacobian %of \ref{eq:jac_4d} %\eqref{eq:model_2d}\eqref{eq:p}\eqref{eq:gamma} 
    $J$ at equilibrium 
    for a range of of $\alpha$ and $\beta$. Computed for 1000 different values of $\alpha$ and $\beta$ over a uniform grid, with $\alpha \in [0.05,0.15] $ and $\beta \in [0.3,2.5]$, for a total of $10^6$ grid points.  The rest of the model parameters are as in Fig. \ref{fig:trajectories}. Oscillations occur in the red region where $\sigma(J)\ge0$.}
    \label{fig:contour}
\end{figure}

A common parametrized path to oscillation in a dynamical system is via a Hopf bifurcation \cite[Theorem 3.4.2]{guckenheimer2013nonlinear}. In a %\textcolor{black}{(supercritical)}
Hopf bifurcation, as a parameter is varied, an equilibrium loses stability and two complex conjugate eigenvalues of its Jacobian simultaneously cross the imaginary axis. At the onset of this instability, oscillations emerge. To find such regions of instability we study numerically the spectral abscissa of \textcolor{black}{\eqref{eq:jac_4d}} for the choice of model parameters $N = M = 1$, $\gamma_0 = 0.05$, $k = 1$, $p_0 = 0.2$, $n = 4$, $\mu = 0.6$, $\nu = 0.45$. The parameters $\alpha$ and $\beta$ are left as free parameters. This choice of parameters satisfied the conditions for equilibrium uniqueness of Theorem \ref{thm:4D-EQ}.

We numerically solve the implicit equations derived in Theorem \ref{thm:4D-EQ} to find the model equilibrium over a grid of $\alpha$, $\beta$, and then evaluate the \textcolor{black}{eigenvalues} of the Jacobian \textcolor{black}{$J$} at each parametrized equilibrium. The result of this numerical study is shown in Fig. \ref{fig:contour}. The parameter values in the gray region of the plot correspond to regions of local stability for the equilibrium. Inside of the red region, the equilibrium is unstable and the Jacobian has a leading pair of complex conjugate eigenvalues. The boundary of the red region corresponds to Hopf bifurcation points, crossing which by varying $\alpha$ and/or $\beta$ from the region exterior into the interior leads to loss of stability and emergence of oscillation. A simulation of a representative oscillation in this region is shown in Fig. \ref{fig:trajectories}(b). Finally, we further confirm the onset of oscillation \textcolor{black}{over} a localized region in $\alpha$ and $\beta$ by computing bifurcation diagrams using numerical continuation software. We fix one of the parameters $\alpha$, $\beta$ and vary the second to compute solution branches and bifurcation points, see Fig. \ref{fig:continuation}. Together these numerical results confirm the existence of parameter regimes in which the flow of ants into and out of the nest oscillates and suggests a mechanistic principle leveraged by turtle ants to enable effective trail optimization.%, which in turn may 

%plugging in our equations to get this in terms of $R^*$ gives
%\begin{equation} 
%  \begin{bmatrix} \alpha(\frac{1-R^*}{R^*}-1) & - \beta \frac{M}{N} R^* & 0 & 0 \\ 
%    -\frac{N^2}{M^2}\frac{1-R^*}{R^*} & -f(2\frac{\nu}{\mu}\alpha  \frac{N^2}{M}(1-R^*)) -\gamma_0 - \beta \frac{N}{M}R^* & 0 & 1-\frac{\alpha}{\beta}\frac{N}{M}\frac{1-R^*}{R^*}\\
 %   \nu \alpha \frac{N^2}{M}\frac{1-R^*}{R^*} & \nu(M f(2\frac{\nu}{\mu}\alpha  \frac{N^2}{M}(1-R^*)) -\gamma_0 M + \beta N R^*) & - \mu & \nu M (1-\frac{\alpha}{\beta}\frac{N}{M}\frac{1-R^*}{R^*}) \\
 %   0 & 0 & f'(2\frac{\nu}{\mu}\alpha  \frac{N^2}{M}(1-R^*))/\tau_{\gamma} & - 1/\tau_{\gamma}\end{bmatrix}\label{eq:jac_4d}
%\end{equation}

\begin{figure}[t]
    \centering
    \includegraphics[width=1.0 \linewidth]{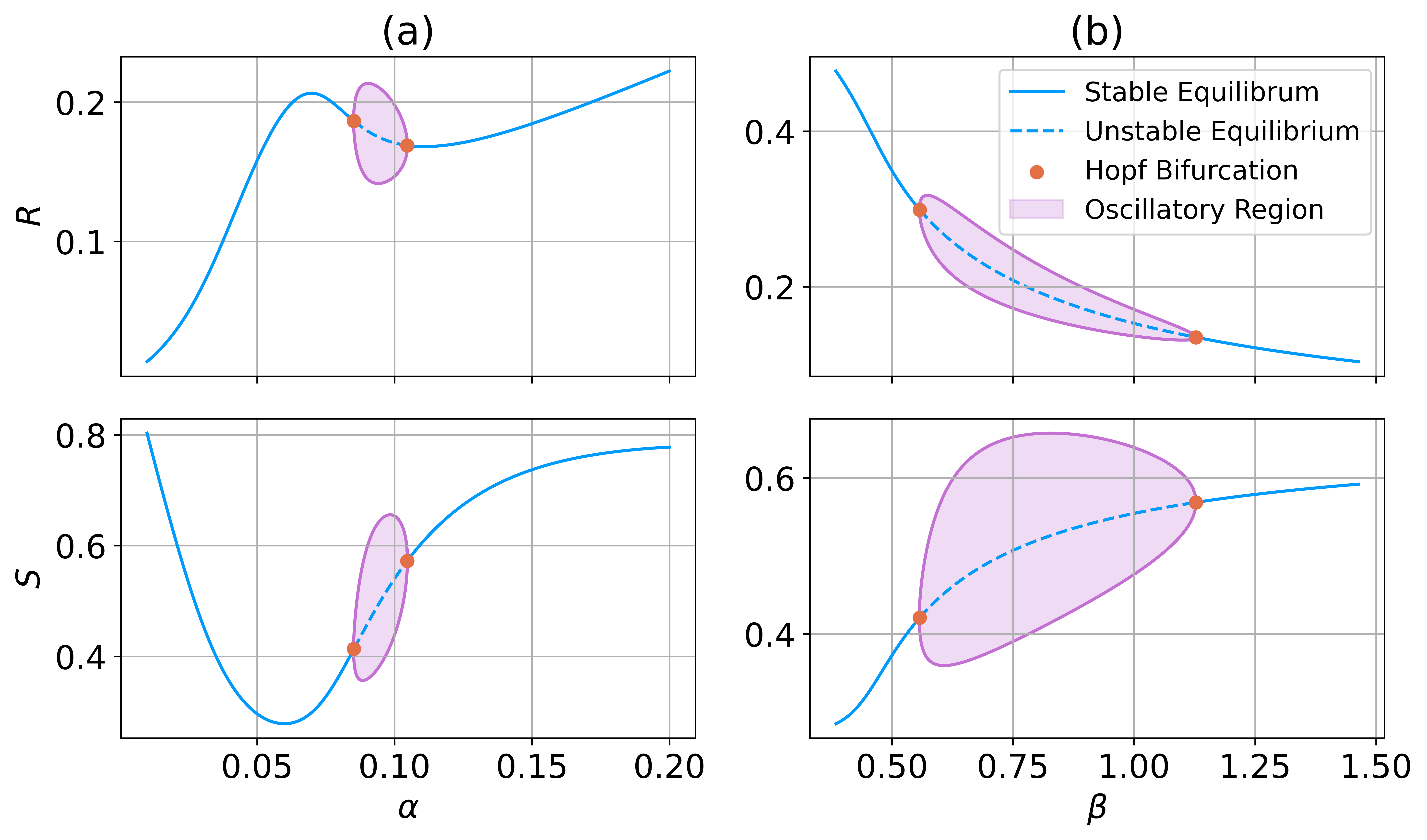}
    \caption{Numerical Continuation of \eqref{eq:model_2d},\eqref{eq:p},\eqref{eq:gamma}, computed with the Julia package BifurcationKit.jl \cite{veltz:hal-02902346}
    \label{fig:continuation} in parameters $\alpha$ (column (a)) and $\beta$ (column (b)), plotted with respect to $R$ and $S$. The vertical width of the oscillatory region at a fixed parameter value corresponds to the range of the limit cycle. Parameters not listed in the plot are \textcolor{black}{the} same as in Fig. \ref{fig:trajectories}.}
\end{figure}

%\begin{theorem} \label{thm:2D-Hopf}
%    Hopf bifurcation + mixed-feedback oscillations
%\end{theorem}
%\begin{proof}
    
%\end{proof}

%List of figures: 
%\begin{itemize}
%    \item Figure 1: (a) Model setup section: diagram showing compartmental model, different transition rates; (b) Shape of saturation function $f(p)$
%    \item Figure 2: (a) Simulation of 2D model settling to an equilibrium (b) simulation of 4D model oscillating at the same parameters 
%    \item Figure 3: (a) continuation in alpha showing Hopf bifurcation; (b) continuation in beta showing hopf bifurcation 
%    \item Figure 4: alpha-beta grid showing localized parameter region where Hopf bifurcation occurs 
%\end{itemize}

\section{Discussion and Future Work\label{sec:discussion}}

In this work we proposed and analyzed a mechanistic model of turtle ant foraging dynamics to explore onset of oscillations in the flow into and out of a nest. We showed that \textcolor{black}{in our modeling approach,} trophallaxis interactions between ants at the nest alone are not enough to induce oscillations in their flow. However when these interactions are considered in combination with dynamic pheromone feedback, oscillations become possible. Our numerical studies reveal a localized region in the $(\alpha, \beta)$ parameter space in which the flow oscillates. Interestingly, we see that the oscillation persists for a wider range of $\beta$ while small variations in $\alpha$ can push the system out of the oscillatory regime. The transition rates $\alpha$ and $\beta$ are measurable in field experiments, which means that our model-based analytic insights can be tested against real-world observations of turtle ants.

In future work, we will extend our modeling by relaxing the simplifying assumption made in this paper that trail ants do not permanently return to the nest, and that nest ants do not go out to forage. We will consider the effects of transitions between nest and trail ants to derive a more general set of conditions for oscillation onset. % will allow us to draw more general insights about foraging dynamics and the emergence of oscillations.

Furthermore, we plan to expand our modeling to consider multiple nests to understand effects of expanding the trail network on the properties of the emergent oscillation.
\textcolor{black}{Additionally, we aim to examine other mechanisms that could give rise to oscillations, such as delays or stochasticity in the foraging process, as well as effects from the geometry of the trail network which was not modeled explicitly here.}
 %in order to understand the effect of adding more nests on the emergent oscillation. 

Our ultimate goal with this research effort is to generate a set of model-based testable hypotheses that can be compared against field data and can be used to design future experiments. \textcolor{black}{Turtle ant foragers returning to the nest interact with other ants outside the nest \cite{gordon2017local}% which presents an opportunity for experimentation not available for other ant species, for which these interactions are not observed outside the nest. %with turtle ants as species of interest
, further work will investigate the effects on the flow of interactions between returning foragers and ants receiving liquid food.}  This work will provide new insights into the role of oscillations in decentralized computation in nature.


\end{document}